\documentclass[lettersize,journal]{IEEEtran}

\usepackage{orcidlink}
\usepackage{amsmath,amssymb,amsfonts,amsthm,bm}
\usepackage{graphicx,booktabs,array,cite,url}
\usepackage{float}
\usepackage[caption=false,font=footnotesize]{subfig}

\newtheorem{theorem}{Theorem}
\newtheorem{proposition}{Proposition}

\theoremstyle{remark}
\newcommand{\CN}{\mathcal{CN}}
\newcommand{\E}{\mathbb E}
\newcommand{\Prob}{\mathbb P}
\newcommand{\tr}{\operatorname{tr}}
\newcommand{\rank}{\operatorname{rank}}
\newcommand{\diag}{\operatorname{diag}}
\newcommand{\range}{\operatorname{range}}
\newcommand{\nullsp}{\operatorname{null}}

\newcommand{\HH}{^{\mathsf H}}
\newcommand{\TT}{^{\mathsf T}}
\newcommand{\norm}[1]{\left\|#1\right\|_2}
\newcommand{\pos}[1]{\left[#1\right]^+}

\usepackage{algorithm}
\usepackage{algpseudocode}  

\usepackage[T1]{fontenc}

\algrenewcommand\algorithmicrequire{\textbf{Input:}}
\algrenewcommand\algorithmicensure{\textbf{Output:}}

\def\BibTeX{{\rm B\kern-.05em{\sc i\kern-.025em b}\kern-.08em
		T\kern-.1667em\lower.7ex\hbox{E}\kern-.125emX}}
\usepackage{balance}
\makeatletter
\newcommand*{\rom}[1]{\expandafter\@slowromancap\romannumeral #1@}
\makeatother
\providecommand{\triangleq}{\stackrel{\triangle}{=}}

\theoremstyle{remark}

\begin{document}
	\title{Group Physical-Layer Authentication and Confidential Beamforming in Non-Terrestrial Networks}
	\author{Parsa~Rajabi~\orcidlink{0009-0005-4645-9444},
		Mohammad~Reza~Abedi~\orcidlink{0000-0003-4114-1339}~\IEEEmembership{Member,~IEEE},
		Paeiz~Azmi~\orcidlink{0000-0001-9736-3462}~\IEEEmembership{Senior~Member,~IEEE},
		Nader~Mokari~\orcidlink{0000-0001-5364-8888}~\IEEEmembership{Senior~Member,~IEEE}
		\thanks{This work has been submitted to the IEEE for possible publication.
			Copyright may be transferred without notice, after which this version may
			no longer be accessible.}
		\thanks{P. Rajabi, M. R. Abedi, P. Azmi, and N. Mokari are with the Department of Electrical and Computer Engineering, Tarbiat Modares University, Tehran, 14115-111, Iran (e-mail: parsa\_rajabi@modares.ac.ir; Mohammadreza\_abedi@modares.ac.ir; pazmi@modares.ac.ir; nader.mokari@modares.ac.ir).}
	}
	
	\markboth{ }%
	{ }
	
	\maketitle
	
	\begin{abstract}
		Coordinated high-altitude platform station (HAPS) networks need to authenticate group responses and protect private data under shared power constraints. This paper studies joint group physical-layer authentication (PLA) and confidential beamforming using a fresh secret tag superimposed on private data. For received attacks independent of the current tag, we derive exact finite-block expressions for false acceptance and legitimate rejection. These expressions translate authentication error targets into a required tag signal-to-interference-plus-noise ratio (SINR). We establish a nonzero asymptotic tag-energy floor and analyze residual frequency offset and peak-power restrictions. The legitimate receiver removes the tag by orthogonal projection, which discards one temporal dimension. For colluding eavesdroppers, we identify the spatial coverage condition under which artificial noise (AN) keeps their SINR bounded as data and AN powers increase together. We jointly optimize the private beams, tag beam, and AN covariance to minimize transmit power under authentication, secrecy, and per-HAPS power constraints. The method combines semidefinite relaxation (SDR), exact private-beam rank reduction, and feasible successive convex approximation. A corrected dual lower bound measures the remaining optimality gap. Numerical results show that increasing the tag length from 16 to 64 symbols reduces the required tag SINR by 7.16 dB. Joint beam design reduces mean transmit power by 27.4\% relative to an optimized fixed-tag-direction baseline. The authentication guarantee requires fresh secret tags and excludes relays of the current response.
	\end{abstract}
	\begin{IEEEkeywords}
		Physical-layer authentication, high-altitude platform station, artificial noise, statistical detection, semidefinite optimization.
	\end{IEEEkeywords}

	\section{Introduction}
	\IEEEPARstart{C}{oordinated} high-altitude platform stations (HAPSs) can transmit a common service identifier and multiple private streams over a shared resource. Their distributed aperture offers spatial flexibility, but it also makes the meaning of authentication important: receiving the correct service tag does not identify every transmitting platform. Moreover, a favorable tag secrecy rate is not an impersonation probability. A resource manager needs a quantitative relation between transmit decisions and the actual verifier, together with an explicit description of the attacks covered by that relation.
	
	Active physical-layer authentication (PLA) embeds authentication information in a transmitted waveform; surveys review tag-based, channel-based, and hardware-feature approaches \cite{xie2021}. Concurrent data and secret-tag transmission was studied in \cite{yu2008}. 
	A detection scheme for cooperative replay attacks was proposed in \cite{liang2025physical}, exploiting noise characteristics in replayed signals.
	Data and artificial noise (AN) precoding and tag-power allocation for authentication and secrecy in massive multiple-input multiple-output (MIMO) unmanned aerial vehicle (UAV) communications were studied in \cite{ref31}. Recent papers consider different combinations of authentication and confidentiality. Tag-based encoding for joint secrecy and authentication in a UAV communication setting, including AN, is studied in \cite{wang2025joint}. Two-way privacy-preserving authentication for UAV-assisted MIMO, using carrier-frequency-offset features and session-key establishment, is considered in \cite{zhang2025enhanced}. AN-aided MIMO PLA under imperfect channel state information (CSI) has also been studied directly \cite{perazzone2021}; the use of AN for authentication is therefore not claimed as new here. Distributed PLA has been surveyed as a collaborative architecture \cite{zhang2024}, while HAPS network reviews provide the non-terrestrial deployment context \cite{kurt2021}. The guarantees analyzed here depend on the spherical-tag experiment and threat assumptions stated below, and should not be inferred from these different mechanisms.
	
	Freshness is a separate issue from the choice of detector. Orthogonal frequency-division multiplexing (OFDM) challenge-response PLA is analyzed in \cite{liu2026use}, which reports that correlation among subchannel responses can enable effective attacks, motivating an explicit statement of the randomness premise. Series-cancellation PLA has also been studied \cite{ref48}. Here, the tag is modeled as a fresh secret spherical response generated from provisioned group-key material. Our false-acceptance result applies only when the current received attack waveform is independent of that response; it does not cover a relay that learns and forwards the current response.
	
	Normalized matched-filter statistics have an established literature, including adaptive normalized matched-filter design \cite{kammoun2018}. We do not claim a new matched filter. Instead, we derive its exact finite-block angular law for the stated spherical-template experiment and connect the legitimate-rejection probability to a transmit-side tag constraint. This calibration depends on the receiver statistic and its threat assumptions, not just on a nominal tag secrecy rate.
	
	AN-based spatial protection is also established \cite{goel2008}, and directional AN covariance design has been studied for multi-eavesdropper secrecy optimization \cite{lima2013}. The present formulation keeps the full covariance observed by colluding receivers and evaluates leakage after granting them the tag and all non-target private streams. This is a deliberately conservative leakage benchmark, not an assertion that an impersonator knows the current secret response.
	
	Minimum-power multicast beamforming is generally nonconvex, and semidefinite relaxation (SDR) is a standard analysis tool for it \cite{sidiropoulos2006}. The design also uses established convex-optimization and sequential-approximation methods \cite{boyd2004,beck2010}. Our claim is limited to the specific combination developed below: calibration of the spherical group-tag verifier, the corresponding sufficient secrecy constraints, and a private-beam rank reduction for this model.
	
	Additional tag-based designs address different signaling and channel settings. Authentication under frequency-selective fading is studied in \cite{ref17}; phase-tag authentication is studied in \cite{ref76}; and authentication in spatial modulation is studied in \cite{ref77}. Superimposed-tag authentication over Nakagami-m block-flat fading is analyzed and theoretical predictions are compared with GNU Radio simulations in \cite{ref49}.
	
	Reconfigurable intelligent surfaces (RISs) have also been incorporated into active authentication. Tag-based authentication for RIS is studied in \cite{ref30}, while RIS-assisted challenge-response authentication with multiple-antenna devices is designed in \cite{ref82}. These settings concern authentication mechanisms and propagation architectures different from coordinated HAPS group delivery.
	
	Backscatter systems provide another application of physical-layer challenge-response mechanisms. Mutual authentication between a radio-frequency source and backscatter devices is studied in \cite{zhang2025authscatter}, while a subsequent study considers challenge-response authentication between ambient backscatter devices \cite{ref47}. Joint covert communication and authentication during millimeter-wave beam alignment are investigated in \cite{ref44}. Incentives for collaborative authentication through centralized governance are addressed in \cite{zhou2025incentive}.
	
	Formal security models and message-authentication constructions provide complementary perspectives. A PLA model over fading channels is analyzed in \cite{xie2018security}, and message authentication with secure channel codes is studied in \cite{chen2018physical}. Classical confidential-message coding \cite{csiszar1978} provides background for the Gaussian secrecy benchmark used here; it does not replace the separate impersonation experiment.

	In our earlier work~\cite{parsaPreviousHAPS}, we jointly studied HAPS
	deployment, user association, and power allocation for collaborative
	group-tag and private-data delivery with AN under colluding and
	non-colluding eavesdropping, and evaluated secrecy spectral efficiency
	(SSE) and user fairness. The present paper addresses a different set of questions:
	how to generate and verify a fresh group tag over a finite block, how to
	control legitimate rejection and false acceptance, and how to jointly design
	confidential private beams and AN under calibrated error constraints
	and a conservative colluding-receiver benchmark. Accordingly, the earlier work focused on network-level resource
	management and performance metrics, whereas this paper focuses on
	calibrated group-response authentication, secrecy-constrained beamforming,
	and deterministic minimum-power design. The simulation code is also available in \cite{hapsCode2026}.
	
	The main contributions are summarized as follows.
	\begin{itemize}
		\item We establish a finite-block relation between group-tag verification and transmit-side authentication constraints. For a fresh spherical tag and a received attack independent of that tag, we characterize false acceptance and derive the legitimate-rejection probability. These results determine the tag SINR required to meet prescribed error targets. We further derive a nonzero asymptotic tag-energy floor, quantify residual frequency-offset sensitivity, and obtain a conservative threshold correction for tags selected under a peak-power restriction.
		\item We formulate a confidential-transmission benchmark that explicitly accounts for tag removal and colluding eavesdroppers. The legitimate receiver removes the tag by orthogonal projection, which discards one temporal dimension. The eavesdroppers are granted the tag and all non-target private streams. For fixed beam and covariance shapes under joint data and AN power scaling, we derive a necessary and sufficient spatial coverage condition for bounded eavesdropper SINR. A fixed-beam specialization also identifies the leakage floor and reduces AN power selection to a scalar feasibility test.
		\item We develop a minimum-power design that jointly determines the private beams, group-tag beam, and AN covariance under authentication, secrecy, and per-HAPS power constraints.
	\end{itemize}
	
	The objective is an analytical resource-design framework, not a replacement for cryptographic provisioning. The group tag establishes possession of an authorized group response under the specified attack model. Individual platform identity, malicious authorized members, and payload-message integrity require additional mechanisms.
	
	\emph{Notation:} $(\cdot)\HH$ and $(\cdot)\TT$ denote Hermitian transpose and transpose, and $(\cdot)^\dagger$ denotes the Moore--Penrose inverse. $\norm{\cdot}$, $\tr(\cdot)$, and $\rank(\cdot)$ denote Euclidean norm, trace, and rank. $\mathbf I_m$ is the $m$-dimensional identity, and $\mathbf 1_U$ is the all-ones vector in $\mathbb R^U$; $\mathbf A\succeq0$ means Hermitian positive semidefinite. $\CN(\boldsymbol\mu,\mathbf C)$ is a circular complex Gaussian law, $\E$ denotes expectation, $j^2=-1$, and $[x]^+=\max(x,0)$. Logarithms without a base in proofs are natural.
	
	\section{System Model}
	\label{sec:model}
	
	This section introduces the coordinated transmission model, channel assumptions, and power constraints. It then defines the fresh group tag, attack model, and verification procedure.
	\subsection{Coordinated Group Transmission}
	Consider one scheduled group on one flat-fading resource block (RB) in Fig.~\ref{fig:system-model}. Its $U$ single-antenna users are served by $K$ HAPSs, each with $M$ transmit antennas; $N=KM$ is the total number of transmit dimensions. Different groups may use orthogonal RBs. All rates below use the bandwidth of this RB. The analysis and optimization are conditional on one channel snapshot; no average channel is substituted for an instantaneous spatial constraint.
	
	\begin{figure}[t]
		\centering
		\includegraphics[width=\columnwidth]{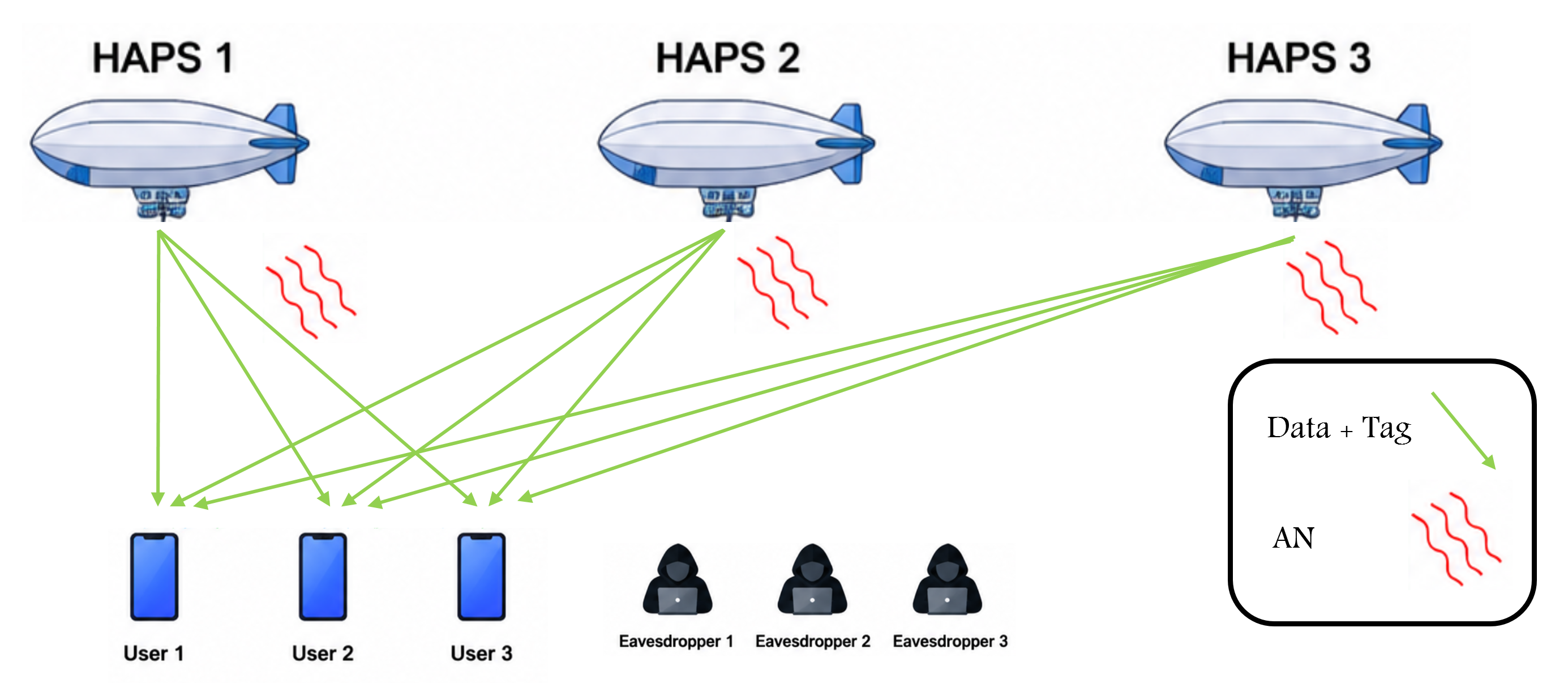}
		\caption{System model.}
		\label{fig:system-model}
	\end{figure}
	
	The stacked channel to user $u$ is $\mathbf h_u\in\mathbb C^N$, with received scalar $\mathbf h_u\HH\mathbf x$. The matrix $\mathbf H=[\mathbf h_1,\ldots,\mathbf h_U]\HH\in\mathbb C^{U\times N}$ collects legitimate channels. Let $\mathbf G\in\mathbb C^{E\times N}$ collect the channels to $E$ cooperating single-antenna eavesdroppers in the same row convention. Their independent receiver noises have variance $\sigma_e^2>0$; legitimate noises have variance $\sigma_u^2>0$. Equal variances are used only in the numerical examples.
	
	During a block of $n\ge2$ complex symbols, the stacked transmitted signal is
	\begin{equation}
		\mathbf x[\ell]=\sum_{u=1}^{U}\mathbf v_u b_u[\ell]
		+\sqrt n\,\mathbf w s[\ell]+\mathbf a[\ell],
		\quad \ell=0,\ldots,n-1,
		\label{eq:tx}
	\end{equation}
	where $\mathbf v_u,\mathbf w\in\mathbb C^N$ are private and shared-tag beams. The private symbols are independent $\CN(0,1)$ variables for the Gaussian signaling model. The temporal tag $\mathbf s\in\mathbb C^n$ has $\norm{\mathbf s}=1$. AN satisfies $\mathbf a[\ell]\sim\CN(\mathbf0,\mathbf Q)$ independently over $\ell$, with $\mathbf Q\succeq0$, independently of the private data and tag. All authorized HAPSs transmit their blocks of the \emph{same} $\mathbf w\mathbf s\TT$. This is coherent collaborative group delivery, not a sum of independent tag powers.

	Partition each network-wide beam into HAPS-specific blocks as
	$\mathbf v_u=[\mathbf v_{u,1}^{\mathsf T},\ldots,
	\mathbf v_{u,K}^{\mathsf T}]^{\mathsf T}$ and
	$\mathbf w=[\mathbf w_1^{\mathsf T},\ldots,\mathbf w_K^{\mathsf T}]^{\mathsf T}$,
	where each block has $M$ entries. HAPS $k$ contributes to user $u$ whenever
	$\mathbf v_{u,k}\neq\mathbf 0$. Thus, users may be served by different
	subsets of HAPSs, and these subsets may overlap. Equation~\eqref{eq:tx}
	represents the stacked network-wide signal; each HAPS transmits only the
	components corresponding to its own antenna block.
	
	Let $\mathbf D_k$ be the diagonal selector for the $M$ antennas of HAPS $k$, and let $P_k^{\max}>0$ denote its radio-frequency (RF) transmit-power budget. Average power over the entire tag block is
	\begin{equation}
		P_k=\sum_u\norm{\mathbf D_k\mathbf v_u}^2
		+\norm{\mathbf D_k\mathbf w}^2+\tr(\mathbf D_k\mathbf Q).
		\label{eq:power}
	\end{equation}
	The normalization $\sqrt n$ in \eqref{eq:tx} makes the tag power $\norm{\mathbf w}^2$, independently of $n$. Individual tag samples need not have constant modulus. Power-amplifier backoff and peak constraints are separate engineering requirements.
	
	We restrict AN to the legitimate-channel null space. Let $d=N-\rank(\mathbf H)$ and let $\mathbf F\in\mathbb C^{N\times d}$ have orthonormal columns spanning $\nullsp(\mathbf H)$. Then
	\begin{equation}
		\mathbf Q=\mathbf F\mathbf A\mathbf F\HH,\qquad
		\mathbf A\succeq0,\qquad \mathbf H\mathbf Q=\mathbf0.
		\label{eq:null}
	\end{equation}
	For $d=0$, $\mathbf Q=\mathbf0$. Unlike independent per-HAPS AN, this full covariance can correlate AN across platforms. A trusted controller must distribute the corresponding random samples or shared random seeds. The full-covariance model therefore assumes coordinated fronthaul as well as shared private data. Local AN can be enforced by additionally requiring $\mathbf Q$ to be block diagonal, but that restricted architecture is not used for the optimized curves.
	
	All effective channels include propagation and calibrated carrier phases. Timing, carrier-frequency compensation, and channel constancy over the block are assumed for the principal analysis. This does not require estimating the aggregate tag amplitude at the verifier. It does require the beamformer to know the channel snapshot used for spatial design. Eavesdropper CSI is a design benchmark, not an assertion that arbitrary passive receivers can always be located or sounded.
	
	\subsection{Freshness, Threat Model, and Verification}
	A trusted initialization supplies a group key $K_g$ to its authorized HAPSs and verifiers. For a fresh challenge $c$, group identifier $g$, and public frame identifier $f$, a keyed pseudorandom function (PRF) expansion produces a vector $\mathbf z\in\mathbb C^n$; its normalized version is the tag:
	\begin{equation}
		\mathbf s=\mathbf z/\norm{\mathbf z},\qquad
		\mathbf z\sim\CN(\mathbf0,\mathbf I_n).
		\label{eq:tag}
	\end{equation}
	Every challenge/frame tuple is accepted at most once. The secure expansion and finite-precision Gaussian mapping are implementation assumptions. The exact continuous laws concern the ideal spherical experiment; they are not information-theoretic security claims for a finite secret key.
	
	Under the authorized hypothesis $\mathcal H_1$, the $n$ samples at user $u$ form
	\begin{equation}
		\begin{aligned}
			\mathbf y_u&=\sqrt n\,a_u\mathbf s+\boldsymbol\xi_u,
			\qquad a_u=\mathbf h_u\HH\mathbf w,\\
			\boldsymbol\xi_u&\sim\CN(\mathbf0,v_u^{\rm a}\mathbf I_n),
		\end{aligned}
		\label{eq:authorized}
	\end{equation}
	where $v_u^{\rm a}$ denotes the aggregate interference-plus-noise
	variance at user $u$, defined as
	\begin{equation}
		v_u^{\rm a}=\sum_j|\mathbf h_u\HH\mathbf v_j|^2
		+\mathbf h_u\HH\mathbf Q\mathbf h_u+\sigma_u^2.
		\label{eq:va}
	\end{equation}
	The superscript $\rm a$ indicates disturbance during authentication; $v_u^{\rm a}$ is a scalar variance, distinct from vector $\mathbf v_u$. Equation~\eqref{eq:authorized} follows exactly for independent Gaussian private symbols, Gaussian AN, and a constant channel. For a coded finite-alphabet waveform it is a modeling approximation that requires separate validation.
	
	Under impersonation $\mathcal H_0$, the authorized response is absent and the received vector $\mathbf y_u$ is independent of the current ideal $\mathbf s$, conditional on public information and design variables. The attacker may choose any transmitted energy, temporal correlation, and waveform using that information and previous transcripts. Several cooperating attackers are allowed if their combined current received waveform still obeys this independence condition. An old recorded response satisfies the condition when the new tag is independent. A relay that observes and forwards the \emph{current} response does not. Neither a short challenge deadline nor a shared key alone proves exclusion of such a relay; it is explicitly outside this experiment.

	For a nonzero received vector, user $u$ computes
	\begin{equation}
		T_u=\frac{|\mathbf s\HH\mathbf y_u|^2}{\norm{\mathbf y_u}^2},
		\qquad \mathcal H_1\ \hbox{is accepted if }T_u\ge t_n.
		\label{eq:test}
	\end{equation}
	Set $T_u=0$ for $\mathbf y_u=\mathbf0$. The test is invariant to an unknown nonzero complex gain and does not need the tag amplitude or disturbance variance to set its threshold. The Gaussian noise model makes a zero observation a probability-zero event. A low-cost received-energy gate may be added, but its legitimate rejection effect would then need separate calibration.
	
	Acceptance authenticates the group response associated with $(c,g,f)$. It does not identify which subset of HAPSs supplied the energy. The shared reference also does not bind each user's private payload unless a verified payload digest is included in the authenticated input. We therefore refer to \emph{group-response authentication and confidential private delivery}, rather than claim a complete message-authentication protocol.
	
	\section{Problem Formulation}
	\label{sec:detection}
	
	We first derive the authentication requirements and examine frequency offset and peak-controlled tags. We then establish the confidentiality model and AN geometry, formulate the minimum-power problem, and develop the solution method and lower bounds.
	\subsection{Authentication Laws and Calibration}
	\begin{theorem}[Fresh-tag false acceptance]
		\label{thm:fa}
		Under $\mathcal H_0$ and the independence experiment above, for any nonzero received waveform and $0\le t\le1$,
		\begin{equation}
			\Prob(T_u\ge t\mid\mathcal H_0)=(1-t)^{n-1}.
			\label{eq:fa}
		\end{equation}
		If the received waveform may be zero, the right-hand side is an upper bound. Consequently, a prescribed per-attempt level $\alpha\in(0,1)$ is obtained by
		\begin{equation}
			t_n=1-\alpha^{1/(n-1)}.
			\label{eq:threshold}
		\end{equation}
	\end{theorem}
	\begin{proof}
		Condition on any nonzero $\mathbf y_u$ and rotate its normalized direction to the first coordinate. Spherical invariance of $\mathbf s$ gives $T_u\overset{d}{=}|z_1|^2/\sum_{i=1}^n|z_i|^2$. The numerator is exponential and the remaining sum is an independent gamma variable of shape $n-1$. Thus $T_u$ is beta distributed with parameters $(1,n-1)$, giving \eqref{eq:fa}. Averaging over any independent waveform distribution preserves the result.
	\end{proof}
	
	The randomness in this proof is the fresh secret template, not an assumption that an adversary emits Gaussian noise. In particular, increasing the amplitude of a blind tone or old response cannot change its angular success probability. This property should not be confused with secrecy of the transmitted tag after observation. If the attacker learns the current direction, choosing $\mathbf y_u=c\mathbf s$ makes $T_u=1$.
	
	For $J$ attempts at possibly correlated receivers, a union bound gives a probability at most $J\alpha$ that at least one is falsely accepted, provided freshness and conditional independence hold at each attempt. No independence among receiver decisions is needed. For a finite implementation, let $\varepsilon_{\rm PRF}$ bound the distinguishing advantage, at the assumed adversarial computational budget, between the keyed expansion and independent random inputs. Let $\varepsilon_{\rm map}$ bound the increase in false acceptance caused by finite-precision Gaussian mapping and normalization in the random-input experiment relative to the ideal spherical experiment, over the allowed attack strategies. Conditional on these bounds, $P_{\rm FA}\le\alpha+\varepsilon_{\rm PRF}+\varepsilon_{\rm map}$ for one experiment. Neither term is quantified by the physical simulations. Nonce reuse or leakage of the current response invalidates the premise instead of merely changing its signal-to-interference-plus-noise ratio (SINR).

	Define the authentication SINR and block noncentrality as
	\begin{equation}
		\rho_u=\frac{|\mathbf h_u\HH\mathbf w|^2}{v_u^{\rm a}},
		\qquad q_u=n\rho_u.
		\label{eq:authsnr}
	\end{equation}
	Let $I_t(a,b)$ denote the regularized incomplete beta function, and let $F_{d_1,d_2,\lambda}^{\rm ncf}$ be the cumulative distribution function (CDF) of a noncentral $F$ variable with degrees of freedom $d_1,d_2$ and numerator noncentrality $\lambda$.
	
	\begin{theorem}[Exact calibration]
		\label{thm:md}
		Under \eqref{eq:authorized}, the legitimate rejection probability is
		\begin{align}
			P_{{\rm MD},u}
			&=e^{-q_u}\sum_{r=0}^{\infty}\frac{q_u^r}{r!}
			I_{t_n}(1+r,n-1) \label{eq:mdseries}\\
			&=F_{2,2(n-1),2q_u}^{\rm ncf}
			\left(\frac{(n-1)t_n}{1-t_n}\right).\label{eq:mdf}
		\end{align}
		For $0<\beta<1-\alpha$, there is a unique positive $q_n^\star$ for which \eqref{eq:mdseries} equals $\beta$. The condition $P_{{\rm MD},u}\le\beta$ is equivalent to
		\begin{equation}
			|\mathbf h_u\HH\mathbf w|^2\ge\rho_n^\star v_u^{\rm a},
			\qquad \rho_n^\star=q_n^\star/n.
			\label{eq:calibrated}
		\end{equation}
	\end{theorem}
	\begin{proof}
		Rotate the first basis vector to $\mathbf s$ and divide by $\sqrt{v_u^{\rm a}}$. The parallel energy is $A=|z_1+\sqrt{q_u}|^2$ after an immaterial phase rotation; the orthogonal energy is $B=\sum_{i=2}^n|z_i|^2$, independently. Then $2A\sim\chi'^2_2(2q_u)$, $2B\sim\chi^2_{2(n-1)}$, and $T_u=A/(A+B)$. The Poisson mixture of the numerator's noncentral chi-square law yields \eqref{eq:mdseries}; the ratio yields \eqref{eq:mdf}. Increasing the noncentrality strictly increases the numerator in stochastic order, while $B>0$ almost surely. Hence rejection decreases continuously from $1-\alpha$ to zero, proving uniqueness and \eqref{eq:calibrated}.
	\end{proof}
	
	Only a scalar root search is needed before spatial optimization. The two probabilities have different roles: $\alpha$ sets the threshold from the attack experiment, while $\beta$ sets the required legitimate tag energy. A minimum received tag SINR without this calibration would leave the actual rejection target unspecified.
	
	\begin{proposition}[Nonzero tag-energy floor]
		\label{prop:energy}
		Fix $\alpha$ and $\beta$ as in Theorem~\ref{thm:md}, and set $a=-\ln\alpha$. As $n\to\infty$,
		\begin{equation}
			q_n^\star\longrightarrow q_\infty^\star,
			\qquad F_{\chi'^2_2(2q_\infty^\star)}(2a)=\beta.
			\label{eq:energyfloor}
		\end{equation}
		Thus $\rho_n^\star$ decreases asymptotically as $q_\infty^\star/n$, but the required normalized received tag energy does not vanish.
	\end{proposition}
	\begin{proof}
		For fixed $q$, $B/(n-1)\to1$ almost surely and $nt_n\to a$. Therefore $nA/(A+B)$ converges in distribution to $A$ and rejection converges to $\Prob(A<a)$. Continuity and strict monotonicity in $q$ bracket the finite-$n$ root between $q_\infty^\star\pm\epsilon$ for all sufficiently large $n$.
	\end{proof}
	
	This is a received-energy statement conditional on disturbance variance. Increasing $n$ cannot provide authentication with arbitrarily little block energy. The result also separates a processing gain from spatial gain: longer tags reduce the required per-symbol ratio, whereas better beams reduce the physical transmit energy needed to realize that ratio.
	
	\subsection{Frequency Offset and Peak-Controlled Tags}
	An uncompensated phase slope changes the tag direction during the block. To isolate this effect, consider a common residual frequency offset $\Delta f$ affecting all received streams, a sample period $T_s$, and constant effective channel magnitudes. Set
	\begin{equation}
		\mathbf D_f=\diag(e^{j2\pi\Delta f\ell T_s})_{\ell=0}^{n-1},
		\quad c_f=|\mathbf s\HH\mathbf D_f\mathbf s|^2.
	\end{equation}
	Circular Gaussian private interference remains temporally white after this phase rotation. Conditional on $\mathbf s$, the desired mean has parallel and orthogonal normalized energies $q_{\parallel}=q_uc_f$ and $q_{\perp}=q_u(1-c_f)$.
	
	\begin{proposition}[Distorted-tag law]
		\label{prop:cfo}
		The exact conditional rejection probability is
		\begin{equation}
			P_{\rm rej}=e^{-(q_{\parallel}+q_{\perp})}
			\sum_{r=0}^{\infty}\sum_{j=0}^{\infty}
			\frac{q_{\parallel}^r q_{\perp}^j}{r!j!}
			I_{t_n}(1+r,n-1+j).
			\label{eq:doublebeta}
		\end{equation}
		Moreover, for an ideal spherical tag,
		\begin{equation}
			\E[c_f]=\frac{n+\left|\sum_{\ell=0}^{n-1}e^{j2\pi\Delta f\ell T_s}\right|^2}{n(n+1)}.
			\label{eq:mean_coherence}
		\end{equation}
	\end{proposition}
	\begin{proof}
		Orthogonal Gaussian components remain independent, but both energy terms now have noncentral chi-square laws. Applying their two independent Poisson mixtures gives \eqref{eq:doublebeta}. For \eqref{eq:mean_coherence}, use $\E|s_i|^4=2/[n(n+1)]$ and $\E(|s_i|^2|s_j|^2)=1/[n(n+1)]$ for $i\ne j$ and expand the squared sum.
	\end{proof}
	
	Equation~\eqref{eq:doublebeta} must be averaged over actual tags for an unconditional prediction; substituting \eqref{eq:mean_coherence} into it is generally not exact. For a fixed distorted template, as tag power alone tends to infinity, $T_u\to c_f$. If $c_f<t_n$, increasing tag power cannot prevent rejection. Frequency alignment is therefore a physical prerequisite, not a problem that can always be repaired by spending more power. Independent HAPS frequency errors can also change the private interference variance with time and require a more general colored-disturbance model.

	\label{sec:peak}
	The spherical construction has fixed block energy but variable sample magnitude. Its tag-component peak-to-average power ratio (PAPR) is
	\begin{equation}
		\mathsf P_s=n\max_{0\le\ell<n}|s[\ell]|^2.
	\end{equation}
	This ratio is shared by the tag component at every transmitting antenna with a nonzero tag weight. It does not bound the peak of the complete waveform, which also contains private data and AN.
	
	\begin{proposition}[Peak restriction with a recalibrated verifier]
		\label{prop:peak}
		For $1<\chi<n$, the probability that an ideal candidate satisfies $\mathsf P_s\le\chi$ is
		\begin{equation}
			p_\chi=\sum_{j=0}^{\lfloor n/\chi\rfloor}
			(-1)^j\binom nj\left(1-\frac{j\chi}{n}\right)^{n-1}.
			\label{eq:papr}
		\end{equation}
		Generate independent candidate tags until the peak condition is met. The expected number of candidates is $1/p_\chi$. For this conditioned tag distribution, the threshold
		\begin{equation}
			t_{n,\chi}=1-(\alpha p_\chi)^{1/(n-1)}
			\label{eq:peakthreshold}
		\end{equation}
		ensures false acceptance at most $\alpha$ against every independent attack waveform. The legitimate rejection law remains \eqref{eq:mdseries} with this new threshold.
	\end{proposition}
	\begin{proof}
		The vector $(|s[0]|^2,\ldots,|s[n-1]|^2)$ is uniform on the probability simplex. For any selected $j$ coordinates, the probability that each exceeds $\chi/n$ is $(1-j\chi/n)^{n-1}$ when the shifted simplex is nonempty, and zero otherwise. Inclusion--exclusion gives \eqref{eq:papr}. If $C$ is the peak-admissibility event, then, for every fixed nonzero received vector,
		\begin{equation}
			\Prob(T_u\ge t\mid C)\le
			\frac{\Prob(T_u\ge t)}{p_\chi}
			=\frac{(1-t)^{n-1}}{p_\chi}.
		\end{equation}
		The independence assumption permits conditioning on the attack vector first. Equation~\eqref{eq:peakthreshold} follows. Under $\mathcal H_1$, the rotation argument in Theorem~\ref{thm:md} holds for \emph{every} fixed unit template, so conditioning on its peak does not change that law except through the threshold.
	\end{proof}
	
	The verifier and authorized platforms reproduce the same candidate sequence using a counter in the keyed expansion. This is rejection sampling of secret tag symbols, not an optimization algorithm. An inexpensive conservative alternative follows from
	\begin{equation}
		p_\chi\ge1-n(1-\chi/n)^{n-1}.
	\end{equation}
	Any strictly positive lower bound on $p_\chi$ can replace it in \eqref{eq:peakthreshold} while preserving the false-acceptance guarantee. Ordinary amplitude clipping is different: it changes the tag distribution and cannot inherit \eqref{eq:fa} unchanged.
	
	Figure~\ref{fig:peak} evaluates the exact simplex formula using high-precision arithmetic to avoid cancellation in its alternating sum. Independent normalized-exponential samples validate the candidate acceptance probability. The SINR penalty uses the same scalar root calibration with $\alpha p_\chi$ in place of $\alpha$. Thus the peak constraint has a quantifiable authentication cost rather than an unspecified implementation penalty. Network figures below use the unconditioned tag; the peak-controlled variant is a separately evaluated extension. Its frequency-offset behavior can still be computed from the actual template in \eqref{eq:doublebeta}, but the unconditioned mean in \eqref{eq:mean_coherence} no longer applies.
	
	\begin{figure}[t]
		\centering
		\subfloat[Acceptance probability.]{\includegraphics[width=.47\columnwidth]{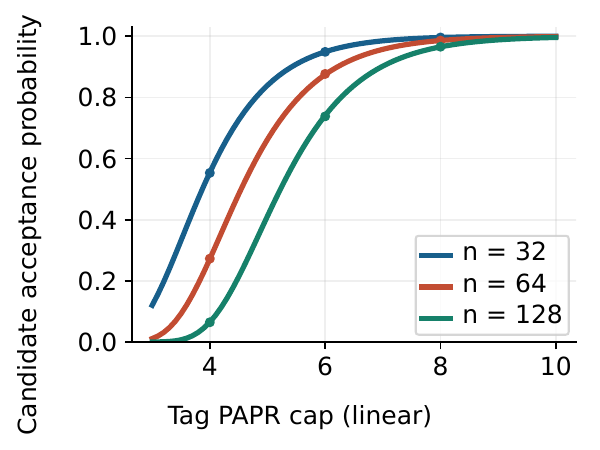}\label{fig:peak:a}}
		\hfill
		\subfloat[Additional required tag SINR.]{\includegraphics[width=.47\columnwidth]{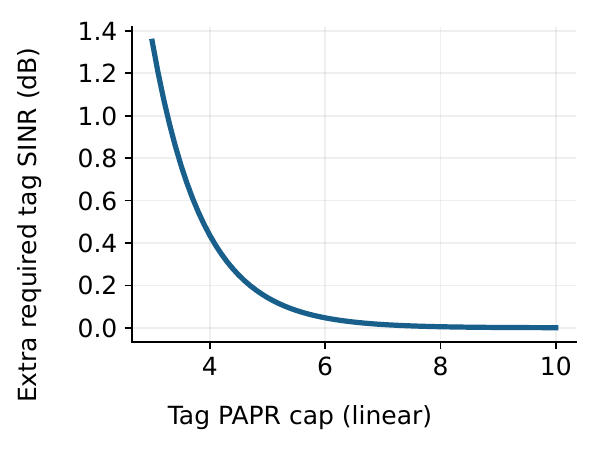}\label{fig:peak:b}}
		\caption{Peak-controlled tag extension. Panel (a) compares the exact probability in \eqref{eq:papr} with $100{,}000$ independent samples per length. In panel (a), solid curves show the analytical probabilities, while circular markers show Monte Carlo estimates. Panel (b) shows the additional required tag SINR for $n=64$, false acceptance at most $10^{-4}$, and rejection $10^{-2}$ after the threshold correction in \eqref{eq:peakthreshold}.}
		\label{fig:peak}
	\end{figure}
	
	\subsection{Confidentiality and AN Geometry}
	\label{sec:secrecy}
	
	Let $\mathbf B_s\in\mathbb C^{n\times(n-1)}$ be any semi-unitary basis satisfying $\mathbf B_s\HH\mathbf s=\mathbf0$. The verifier can form $\widetilde{\mathbf y}_u=\mathbf B_s\HH\mathbf y_u$ before or after its acceptance decision. Since
	\begin{equation}
		\mathbf B_s\HH\mathbf B_s=\mathbf I_{n-1},\qquad
		\mathbf B_s\mathbf B_s\HH=\mathbf I_n-\mathbf s\mathbf s\HH,
		\label{eq:projection}
	\end{equation}
	the tag disappears exactly under block coherence, without knowing $a_u$. A unitary rotation preserves independent Gaussian private symbols and noise in the remaining $n-1$ dimensions. Treating the other private streams as noise yields
	\begin{equation}
		\gamma_u=\frac{|\mathbf h_u\HH\mathbf v_u|^2}
		{\sum_{j\ne u}|\mathbf h_u\HH\mathbf v_j|^2+
			\mathbf h_u\HH\mathbf Q\mathbf h_u+\sigma_u^2}.
		\label{eq:legit_sinr}
	\end{equation}
	The nominal Gaussian spectral efficiency is $c_n\log_2(1+\gamma_u)$, where $c_n=(n-1)/n$. This factor accounts for the discarded data dimension; the receiver is not credited with perfect, free tag cancellation. Under \eqref{eq:null}, the AN term in \eqref{eq:legit_sinr} is zero. The projection loss and the sufficient rate target are detailed in Appendix~\ref{app:projection-rate}.

	For the secrecy calculation, grant the eavesdroppers $\mathbf s$, its effective channel, and all private streams other than the target user's. They can cancel these quantities and jointly combine their $E$ observations. Define
	\begin{equation}
		\begin{aligned}
			\mathbf f_u&=\mathbf G\mathbf v_u,\qquad
			\gamma_{e,u}=\mathbf f_u\HH\mathbf C_e^{-1}\mathbf f_u,\\
			\mathbf C_e&=\mathbf G\mathbf Q\mathbf G\HH+\sigma_e^2\mathbf I_E.
		\end{aligned}
		\label{eq:eve}
	\end{equation}
	The full covariance retains correlations across eavesdroppers. Unlike the legitimate receiver, the genie-aided receiver retains all $n$ dimensions. Accordingly, we use
	\begin{equation}
		\mathcal R_u=\pos{c_n\log_2(1+\gamma_u)-\log_2(1+\gamma_{e,u})}
		\label{eq:secrecy_metric}
	\end{equation}
	as the conservative Gaussian wiretap-rate benchmark. The underlying mutual-information difference is the standard random-coding quantity \cite{csiszar1978}; it is not the secrecy capacity of this multiuser network. The genie can only increase leakage relative to a receiver that must estimate or decode the tag and other users itself.
	
	The finite length $n$ is used exactly in tag detection and projection. Equation~\eqref{eq:secrecy_metric} remains an asymptotic Gaussian coding metric per RB symbol, evaluated on the projected channel. It does not predict the error probability of a particular length-$n$ data code. Authentication-dependent selection of blocks also precludes interpreting $(1-P_{\rm MD})\mathcal R_u$ as an established capacity. These distinctions avoid conflating a finite-block detector with a finite-block data-coding theorem.

	Consider a fixed beam/covariance shape and a joint power scale $p$: $\mathbf v_u=\sqrt p\,\overline{\mathbf v}_u$ and $\mathbf Q=p\overline{\mathbf Q}$. Set $\mathbf f=\mathbf G\overline{\mathbf v}_u$, $\mathbf A_e=\mathbf G\overline{\mathbf Q}\mathbf G\HH$, and let $\mathbf\Pi_0$ project onto $\nullsp(\mathbf A_e)$.
	
	\begin{theorem}[AN rank and orientation criterion]
		\label{thm:rank}
		As $p\to\infty$,
		\begin{equation}
			\gamma_{e,u}(p)=\frac{p}{\sigma_e^2}\norm{\mathbf\Pi_0\mathbf f}^2
			+\mathbf f\HH\mathbf A_e^{\dagger}\mathbf f+O(p^{-1}).
			\label{eq:rank_law}
		\end{equation}
		The colluding SINR remains bounded exactly when $\mathbf f\in\range(\mathbf A_e)$. A full-rank $\mathbf A_e$ is sufficient but not necessary.
	\end{theorem}
	\begin{proof}
		Diagonalize $\mathbf A_e=\mathbf U\diag(a_i)\mathbf U\HH$, $a_i\ge0$, and put $\mathbf z=\mathbf U\HH\mathbf f$. Then
		\begin{equation}
			\gamma_{e,u}(p)=\sum_{i:a_i>0}\frac{p|z_i|^2}{pa_i+\sigma_e^2}
			+\frac{p}{\sigma_e^2}\sum_{i:a_i=0}|z_i|^2.
			\label{eq:spectral}
		\end{equation}
		Expand each positive-eigenvalue term; the zero-eigenvalue terms are already exact. The range condition is equivalent to all zero-eigenvalue coefficients being zero.
	\end{proof}
	
	If legitimate multiuser interference and AN are nulled, $\gamma_u$ grows proportionally to $p$. When an exposed Eve component exists, the two logarithms in \eqref{eq:secrecy_metric} have prelogs $c_n$ and one. The positive-part secrecy metric eventually becomes zero. When the desired direction is covered, the Eve term saturates and the legitimate logarithm can grow. This asymptotic statement applies to the stated fixed-shape power scaling, not to a redesign that changes its spatial null spaces with $p$.
	
	Since $\rank(\mathbf A_e)\le\min(E,d,\rank(\mathbf A))$, insufficient AN dimension creates a potential escape subspace. Dimension counting alone does not prove that the target signal occupies it. Conversely, many transmit antennas do not establish protection unless their AN images cover the desired Eve direction. For example, in six Eve dimensions, a rank-four AN covariance can either protect or expose the desired stream depending on its orientation.

	The following specialization makes the distinction operational. Consider one legitimate user with fixed unit-norm data and tag beams $\mathbf v_0,\mathbf w_0$, a unit-trace nulling AN shape $\mathbf Q_0$, and powers $(p_b,p_t,p_a)$. Let $g_b=|\mathbf h\HH\mathbf v_0|^2>0$, $g_t=|\mathbf h\HH\mathbf w_0|^2>0$, $\mathbf f_0=\mathbf G\mathbf v_0$, and $\mathbf A_0=\mathbf G\mathbf Q_0\mathbf G\HH$. 
	Here, $\mathbf\Pi_0$ denotes the orthogonal projector onto
	$\nullsp(\mathbf A_0)$.
	For a fixed Eve cap $\zeta>0$ and rate requirement $r_0>0$, define
	\begin{equation}
		\Gamma(n,\zeta)=2^{[r_0+\log_2(1+\zeta)]/c_n}-1.
		\label{eq:Gamma}
	\end{equation}
	For these fixed shapes, minimum total power subject to the calibrated authentication constraint and the separate secrecy constraints $\gamma_u\ge\Gamma(n,\zeta)$ and $\gamma_{e,u}\le\zeta$ is attained at
	\begin{equation}
		p_b^\star=\frac{\Gamma\sigma_u^2}{g_b},\qquad
		p_t^\star=\frac{\rho_n^\star(p_b^\star g_b+\sigma_u^2)}{g_t},
		\label{eq:fixed_powers}
	\end{equation}
	and the smallest $p_a\ge0$ for which
	\begin{equation}
		p_b^\star\mathbf f_0\HH(p_a\mathbf A_0+\sigma_e^2\mathbf I)^{-1}\mathbf f_0\le\zeta.
		\label{eq:an_scalar}
	\end{equation}
	The left side is continuous and nonincreasing, and strictly decreasing if the desired vector has a component in the positive-eigenvalue subspace. Its limiting floor is $p_b^\star\norm{\mathbf\Pi_0\mathbf f_0}^2/\sigma_e^2$. A cap below that floor is impossible for these shapes. Equality to the floor requires infinite AN unless the positive-eigenvalue component is zero; a strict cap above the floor has a finite scalar solution whenever zero AN does not already suffice. These statements follow from \eqref{eq:spectral} with data power held fixed. Increasing data power above its required minimum can only increase the required tag power and tighten the leakage constraint, proving \eqref{eq:fixed_powers}. Thus a one-dimensional bisection decides this reduced design. Its optimum need not minimize power under the original secrecy-rate constraint in \eqref{eq:secrecy_metric}. Appendix~\ref{app:fixed-beam} derives the fixed-beam powers and the AN leakage floor.
	
	\subsection{Minimum-Power Beamforming Problem}
	\label{sec:design}
	
	Fix $(n,\alpha,\beta)$, a channel snapshot, and $\zeta>0$. Require $\mathcal R_u\ge r_u>0$ by the sufficient constraints $\gamma_{e,u}\le\zeta$ and $\gamma_u\ge\Gamma_u$, where \eqref{eq:Gamma} uses $r_u$. The design is
	\begin{align}
		\min_{\{\mathbf v_u\},\mathbf w,\mathbf A\succeq0}\quad&
		\sum_u\norm{\mathbf v_u}^2+\norm{\mathbf w}^2+\tr(\mathbf A)
		\label{eq:problem}\\[-3pt]
		\text{subject to}\quad& P_k\le P_k^{\max},\quad k=1,\ldots,K,\nonumber\\
		& |\mathbf h_u\HH\mathbf v_u|^2\ge\Gamma_u
		\left(\sum_{j\ne u}|\mathbf h_u\HH\mathbf v_j|^2+\sigma_u^2\right),\nonumber\\
		&|\mathbf h_u\HH\mathbf w|^2\ge\rho_n^\star
		\left(\sum_j|\mathbf h_u\HH\mathbf v_j|^2+\sigma_u^2\right),\nonumber\\
		&\mathbf v_u\HH\mathbf G\HH\mathbf C_e^{-1}\mathbf G\mathbf v_u\le\zeta,
		\quad u=1,\ldots,U.\nonumber
	\end{align}
	The covariance obeys \eqref{eq:null}. All users in this scheduled group are required to meet their constraints; dropping difficult users is not a hidden source of gains. The objective is physical RF power, with no distance regularizer or tag-concentration proxy. It excludes propulsion and fixed avionics power.
	
	The cap $\zeta$ parameterizes a conservative slice of the secrecy region. Searching a finite cap grid can improve the solution, but a bound for one slice is not a lower bound for the unrestricted secrecy-rate problem. Similarly, $n$ describes an authentication-duration versus transmit-power tradeoff. We report this tradeoff without asserting that the largest $n$ minimizes packet energy or latency.
	
	\subsection{Semidefinite Relaxation and Beam Recovery}
	Let $\mathbf H_u=\mathbf h_u\mathbf h_u\HH$, introduce $\mathbf X_u=\mathbf v_u\mathbf v_u\HH$ and $\mathbf W=\mathbf w\mathbf w\HH$, and temporarily drop their rank-one constraints. The data and authentication inequalities become
	\begin{align}
		\tr(\mathbf H_u\mathbf X_u)&\ge\Gamma_u
		\left(\sum_{j\ne u}\tr(\mathbf H_u\mathbf X_j)+\sigma_u^2\right),
		\label{eq:sdr_data}\\
		\tr(\mathbf H_u\mathbf W)&\ge\rho_n^\star
		\left(\sum_j\tr(\mathbf H_u\mathbf X_j)+\sigma_u^2\right).
		\label{eq:sdr_tag}
	\end{align}
	For a rank-one private matrix, the leakage cap is equivalent to
	\begin{equation}
		\mathbf G\mathbf X_u\mathbf G\HH\preceq
		\zeta(\mathbf G\mathbf F\mathbf A\mathbf F\HH\mathbf G\HH+
		\sigma_e^2\mathbf I_E).
		\label{eq:sdr_eve}
	\end{equation}
	Indeed, whitening by $\mathbf C_e^{-1/2}$ makes the left side rank one with sole nonzero eigenvalue $\gamma_{e,u}$. Retaining \eqref{eq:sdr_eve} for unrestricted $\mathbf X_u\succeq0$ defines the relaxation. Its objective and per-HAPS power constraints are linear traces, so it is a semidefinite program (SDP). The rank-one equivalence of the Eve matrix constraint is detailed in Appendix~\ref{app:eve-lmi}.
	
	The private-beam recovery below uses the rank-one covariance
	construction employed in joint radar--communication beamforming
	\cite{liu2020joint}. The construction preserves the intended
	received power while reducing the covariance in the
	positive-semidefinite order. We verify its applicability to the
	present problem, including the calibrated authentication
	constraints, colluding-eavesdropper leakage caps, and per-HAPS
	power limits. The shared-tag covariance remains subject to
	separate rank recovery.
	
	\begin{theorem}[Only the tag needs rank recovery]
		\label{thm:reduction}
		Given a feasible relaxed point with positive target rates, define
		\begin{equation}
			\widehat{\mathbf v}_u=
			\frac{\mathbf X_u\mathbf h_u}{\sqrt{\mathbf h_u\HH\mathbf X_u\mathbf h_u}}.
			\label{eq:rank_reduce}
		\end{equation}
		Replacing every $\mathbf X_u$ by $\widehat{\mathbf v}_u\widehat{\mathbf v}_u\HH$ preserves feasibility of the relaxation and does not increase power. If $\mathbf W$ is rank one, this yields a globally optimal physical solution of the fixed-$(n,\zeta)$ problem whenever the relaxed point is optimal.
	\end{theorem}
	\begin{proof}
		The denominator is positive by \eqref{eq:sdr_data}. Let $\mathbf t=\mathbf X_u^{1/2}\mathbf h_u/\sqrt{\mathbf h_u\HH\mathbf X_u\mathbf h_u}$, so $\norm{\mathbf t}=1$. Then
		\begin{equation}
			\mathbf\Delta_u\triangleq\mathbf X_u-
			\widehat{\mathbf v}_u\widehat{\mathbf v}_u\HH
			=\mathbf X_u^{1/2}(\mathbf I-\mathbf t\mathbf t\HH)\mathbf X_u^{1/2}\succeq0,
		\end{equation}
		and $\mathbf h_u\HH\mathbf\Delta_u\mathbf h_u=0$. The intended power is unchanged. Every interference power, tag disturbance, and per-HAPS power weakly decreases, since their coefficient matrices are positive semidefinite. The left side of \eqref{eq:sdr_eve} also decreases in matrix order. Applying this argument to all users proves the claim. When the tag is rank one, the recovered point is feasible for the original problem and meets its relaxed lower bound.
	\end{proof}
	
	This reduction is a model-specific consequence of the positive-semidefinite ordering of the calibrated constraints; its applicability depends on the direction of \emph{every} such constraint. In particular, it would fail if the model imposed a lower bound on interference or treated a private covariance reduction as additional useful authentication energy. The theorem does not assert rank-one optimality of a multicast tag with arbitrary user and power constraints.

	For a higher-rank $\mathbf W$, use its dominant eigenvector and a fixed Fourier-phase list in its positive eigenspace as tag directions. This is a deterministic initialization list, not stochastic search. For each nonzero candidate direction $\mathbf d$ and recovered private beams, set
	\begin{equation}
		\mathbf w=\sqrt{\max_u
			\frac{\rho_n^\star(\sum_j|\mathbf h_u\HH\widehat{\mathbf v}_j|^2+\sigma_u^2)}
			{|\mathbf h_u\HH\mathbf d|^2}}\,\mathbf d.
		\label{eq:scale}
	\end{equation}
	Discard directions with a zero useful response or a violated power budget. Among the remaining directions choose the lowest-power point. If the list fails, solve the fixed-tag conic design described in Section~\ref{sec:results} to obtain a second deterministic initialization. If both fail, report recovery failure; this does not prove infeasibility of the nonconvex problem.
	
	Starting from a feasible tag $\mathbf w^{(i)}$, replace its desired power by the affine lower support
	\begin{equation}
		\ell_u(\mathbf w;\mathbf w^{(i)})=
		2\operatorname{Re}\{(\mathbf h_u\HH\mathbf w^{(i)})^*\mathbf h_u\HH\mathbf w\}
		-|\mathbf h_u\HH\mathbf w^{(i)}|^2.
		\label{eq:affine}
	\end{equation}
	Enforce $\ell_u\ge\rho_n^\star(\sum_j|\mathbf h_u\HH\mathbf v_j|^2+\sigma_u^2)$. Each private beam can be phase rotated without changing any constraint, so impose $\operatorname{Im}(\mathbf h_u\HH\mathbf v_u)=0$ and
	\begin{equation}
		\left\|[\{\mathbf h_u\HH\mathbf v_j\}_{j\ne u},\sigma_u]\right\|_2
		\le\frac{\operatorname{Re}(\mathbf h_u\HH\mathbf v_u)}{\sqrt{\Gamma_u}}.
		\label{eq:soc}
	\end{equation}
	The leakage constraint has the Schur form
	\begin{equation}
		\begin{bmatrix}
			\zeta\mathbf C_e&\mathbf G\mathbf v_u\\
			\mathbf v_u\HH\mathbf G\HH&1
		\end{bmatrix}\succeq0.
		\label{eq:schur}
	\end{equation}
	Together with convex power bounds these define a convex conic subproblem. Since $|\mathbf h_u\HH\mathbf w|^2-\ell_u=|\mathbf h_u\HH(\mathbf w-\mathbf w^{(i)})|^2\ge0$, every accepted subproblem solution remains feasible for the original constraints. The previous point is feasible for the next subproblem, so exact solutions give a nonincreasing power sequence bounded below by zero. This proves convergence of objective values, not global convergence or convergence to a unique beam. Numerical implementations must recheck the original expressions before accepting a point.
	
	\begin{algorithm}[t]
		\caption{Calibrated deterministic design for one layout}
		\label{alg:design}
		\small
		\textbf{Input:} Channels, budgets, $(n,\alpha,\beta,\zeta)$, rates $r_u$.
		\begin{enumerate}
			\item Compute $t_n$, $\rho_n^\star$, $\Gamma_u$, and null basis $\mathbf F$.
			\item Solve the SDR \eqref{eq:sdr_data}--\eqref{eq:sdr_eve} and obtain the lower bound \eqref{eq:dualbound}. Report slice infeasibility only with a valid conic certificate.
			\item Apply \eqref{eq:rank_reduce}, \eqref{eq:scale}, and, if needed, the fixed-tag initialization. If no feasible point is found, report recovery failure and retain the lower bound.
			\item Until the power improvement is below tolerance, solve the inner problem \eqref{eq:affine}--\eqref{eq:schur}; check the original constraints and retain a feasible non-worsening point.
			\item Return beams, AN covariance, physical power, and the lower-bound gap.
		\end{enumerate}
	\end{algorithm}
	
	\subsection{Lower Bounds and Deployment Catalog}
	An SDP solver's primal objective is not itself a mathematically certified lower bound when solved inaccurately. A simple correction makes the reported comparison more transparent. Let $\lambda_u,\mu_u,\varphi_k\ge0$ multiply respectively \eqref{eq:sdr_data}, \eqref{eq:sdr_tag}, and the power constraints; let $\mathbf S_u\succeq0$ multiply the Eve matrix inequality. Put $\mathbf D=\sum_k\varphi_k\mathbf D_k$ and
	\begin{align}
		\mathbf R_W&=\mathbf I+\mathbf D-\sum_u\mu_u\mathbf H_u,\label{eq:dualW}\\
		\mathbf R_{X_u}&=\mathbf I+\mathbf D-\lambda_u\mathbf H_u
		+\sum_{j\ne u}\lambda_j\Gamma_j\mathbf H_j\nonumber\\
		&\quad+\rho_n^\star\sum_j\mu_j\mathbf H_j+\mathbf G\HH\mathbf S_u\mathbf G,\\
		\mathbf R_A&=\mathbf F\HH\left(\mathbf I+\mathbf D-
		\zeta\sum_u\mathbf G\HH\mathbf S_u\mathbf G\right)\mathbf F.
		\label{eq:dualmatrices}
	\end{align}
	The constant part of the Lagrangian is
	\begin{equation}
		\begin{aligned}
			d_0&=\sum_u(\lambda_u\Gamma_u+\mu_u\rho_n^\star)\sigma_u^2\\
			&\quad-\zeta\sigma_e^2\sum_u\tr(\mathbf S_u)-\sum_k\varphi_kP_k^{\max}.
		\end{aligned}
	\end{equation}
	Let $\delta\ge0$ be at least the magnitude of the most negative eigenvalue among the three types of matrices in \eqref{eq:dualW}--\eqref{eq:dualmatrices}. Then
	\begin{equation}
		P_{\rm LB}=\max\left\{0,d_0-\delta\sum_kP_k^{\max}\right\}
		\le P_{\rm optimum}.
		\label{eq:dualbound}
	\end{equation}
	To see this, each variable covariance is positive semidefinite and their total trace is at most $\sum_kP_k^{\max}$. The Lagrangian is bounded below by $d_0-\delta\sum_kP_k^{\max}$ and is no larger than the objective at a primal feasible point. This argument allows nonzero stationarity residuals. In computation, scalar multipliers are clipped nonnegative, matrix multipliers are projected onto the positive semidefinite cone, and eigenvalues are recomputed. Floating-point eigenvalue estimates remain numerical certificates rather than interval-arithmetic proofs.

	Let $\mathfrak D$ be a finite catalog of admissible HAPS layouts. Each entry specifies positions that already satisfy altitude, separation, and any imposed displacement limits. Recompute the physical channels and solve Algorithm~\ref{alg:design} for each entry. For fixed $(n,\zeta)$, the smallest feasible power is an upper bound on the catalog optimum, and $\min_{j\in\mathfrak D}P_{{\rm LB},j}$ is a lower bound. Their difference quantifies the remaining catalog design uncertainty. This is not a claim of global optimization over continuous three-dimensional flight paths.
	
	The finite calibration uses a scalar noncentral-$F$ evaluation and root search. The SDR has $(U+1)N^2+d^2$ real Hermitian degrees of freedom before structural elimination, $U$ Eve matrix inequalities of order $E$, and the covariance cones. The refinement uses $2N(U+1)+d^2$ real variables before phase constraints, with $U$ cones of order $E+1$. Dense factorization and eigendecomposition costs grow cubically in the corresponding matrix dimensions; actual solver cost depends on sparsity and iteration count. A catalog multiplies the number of conic solves by its number of entries. No claim of linear online complexity is made, and there is no offline training cost.
	
	\section{Simulation}
	\label{sec:results}
	
	This section specifies the simulation settings and baselines, validates the authentication analysis, and examines frequency-offset sensitivity and AN geometry. It then evaluates transmit power and layout selection and discusses the limitations.
	\subsection{Models, Settings, and Baselines}
	All experiments use newly generated synthetic data. Waveform simulations test the detector independently of the optimization code. We share our code in \cite{hapsCode2026}. Network experiments use ten seeds, $100$--$109$, with identical channels for all compared methods within each seed. Every scheduled group contains three users; the results do not assert large-network scalability. Table~\ref{tab:settings} specifies the settings.
	In the stopping criterion, $P_{\mathrm obj}$ denotes the current total transmit power, i.e., the objective value in \eqref{eq:problem}.
	
	\begin{table}[t]
		\centering\caption{Numerical settings; one active group RB}
		\label{tab:settings}\small
		\begin{tabular}{@{}p{.55\columnwidth}p{.39\columnwidth}@{}}\toprule
			Quantity&Value\\\midrule
			HAPSs / antennas per HAPS / users&$3/4/3$\\
			Colluding single-antenna receivers&$E=8$\\
			Altitude; default ring radius&$20$ km; $4$ km\\
			Candidate ring radii&$2,4,6$ km\\
			User and Eve coordinates&Uniform in $[-4,4]^2$ km$^2$\\
			Carrier wavelength; antenna spacing&$0.15$ m; $0.075$ m\\
			Rician factor; shadowing&$10$ linear; none\\
			Receiver noise per RB&$4\times10^{-14}$ W\\
			Per-HAPS RF power limit&$10$ W\\
			Authentication $(\alpha,\beta)$&$(10^{-4},10^{-2})$\\
			Tag length in resource tests&$16,32,64,128$\\
			Private secrecy-rate requirement&$0.5$ bit/s/Hz per user\\
			Eve SINR cap $\zeta$&$0.3$ (unless varied)\\
			Network seeds&$100$--$109$\\
			Waveform seed&$261019$\\
			Refinement limit; stopping threshold&$12$; $10^{-5}\max(1,P_{\mathrm obj})$\\\bottomrule
		\end{tabular}
	\end{table}
	
	HAPSs lie at equal angles on the specified horizontal ring. Let $\lambda_0$ denote the carrier wavelength. With distance $d_{k,r}$ to a ground receiver $r$, the power gain is $\beta_{k,r}=(\lambda_0/(4\pi d_{k,r}))^2$. A horizontal uniform linear array has unit-modulus steering entries
	\begin{equation}
		[\mathbf a_{k,r}]_m=
		e^{-j2\pi d_{k,r}/\lambda_0}
		e^{-j\pi m(x_r-x_k)/d_{k,r}}.
	\end{equation}
	Here $m=0,\ldots,M-1$, and $x_r$ and $x_k$ are the Cartesian $x$-coordinates of receiver $r$ and HAPS $k$, respectively. The physical row-channel block is
	\begin{equation}
		\sqrt{\beta_{k,r}}\left(\sqrt{\frac{\kappa}{1+\kappa}}\mathbf a_{k,r}\TT
		+\sqrt{\frac1{1+\kappa}}\mathbf z_{k,r}\TT\right),
	\end{equation}
	where $\kappa$ is the linear Rician factor and the independent scattering vectors satisfy $\mathbf z_{k,r}\sim\CN(\mathbf0,\mathbf I_M)$. This is a transparent, conditioned line-of-sight Rician benchmark. Atmospheric attenuation, correlated shadowing, polarization, antenna element patterns, and CSI estimation are not included. Channels are divided by the noise standard deviation for numerical conditioning; optimized powers retain units of watts.
	
	The strong \emph{fixed-tag} baseline fixes only the tag direction to normalized $\mathbf H\HH\mathbf 1_U$. It optimizes all private beams, AN covariance, and tag power under the same constraints through a convex conic problem. Thus the comparison isolates multicast tag design instead of confounding it with an arbitrary power split. An additional \emph{fixed-beam} reference uses zero-forcing private directions, the same tag direction, and isotropic null-space AN, optimizing their scalar powers. A \emph{no-AN} ablation fixes $\mathbf Q=\mathbf0$ and reoptimizes the beams.
	
	The implementation uses $32$ Fourier-phase candidates in addition to the dominant-eigenvector candidate. Numerical tag rank counts eigenvalues of $\mathbf W$ above $\max\{10^{-5}\lambda_{\max}(\mathbf W),10^{-9}\}$, where $\lambda_{\max}(\mathbf W)$ is its largest eigenvalue. The amplitude scale in \eqref{eq:scale} is multiplied by $1+3\times10^{-7}$ as a numerical feasibility margin.
	
	\subsection{Authentication and Waveform Results}
	Figure~\ref{fig:detection} compares \eqref{eq:mdf} with $120{,}000$ independently simulated received blocks per plotted operating point. Error bars are binomial Wilson intervals at nominal $95\%$ confidence. The simulation rotates the known template to the first coordinate and generates a full complex Gaussian received vector; it does not sample the noncentral-$F$ law directly. Separately, the Poisson-beta series and noncentral-$F$ evaluator are checked against each other.
	
	For $\alpha=10^{-4}$ and $\beta=10^{-2}$, required tag SINRs for $n=16,32,64,128$ are $1.8614$, $0.77685$, $0.35832$, and $0.17246$, respectively. Moving from $16$ to $64$ symbols reduces required per-symbol SINR by approximately $7.16$ dB. However, normalized block energy decreases only from $29.7825$ to $22.9326$, approaching the nonzero limit $q_\infty^\star=21.2805$. Thus the processing gain cannot be interpreted as an equal reduction in block tag energy.
	
	\begin{figure}[t]
		\centering
		\subfloat[Legitimate rejection probability.]{\includegraphics[width=.47\columnwidth]{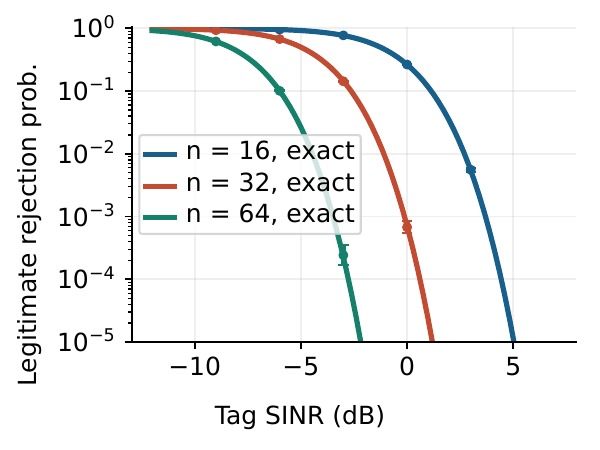}\label{fig:detection:a}}
		\hfill
		\subfloat[Required normalized tag energy.]{\includegraphics[width=.47\columnwidth]{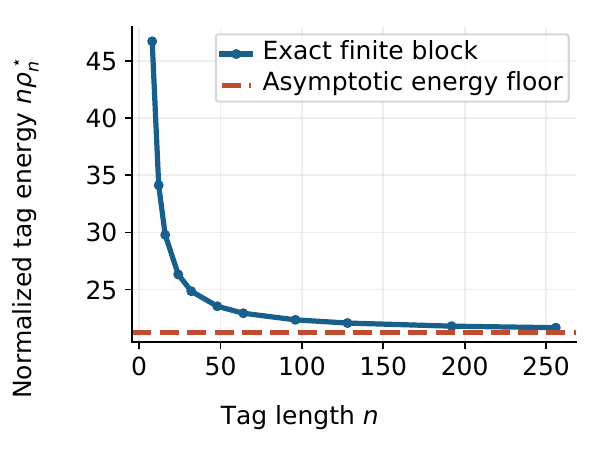}\label{fig:detection:b}}
		\caption{Exact authentication calibration. Panel (a) compares \eqref{eq:mdf} with independent waveform Monte Carlo at false acceptance $10^{-4}$. In panel (a), solid curves show the analytical rejection probabilities, while circular markers show Monte Carlo estimates with 95\% Wilson confidence intervals. Panel (b) shows $n\rho_n^\star$ for rejection $10^{-2}$ and the limit in \eqref{eq:energyfloor}. Zero observed rare-event counts are not plotted as zero on the logarithmic axis. In panel (b), circular markers denote analytical values computed from the required-SINR relation, not Monte Carlo estimates.}
		\label{fig:detection}
	\end{figure}
	
	Figure~\ref{fig:attacks} uses impulse, constant-tone, independently generated old-tag, and temporally colored Gaussian attack signals. Each is tested at attack-to-noise ratios $-10$ and $20$ dB with $300{,}000$ independent fresh templates. Here $\alpha=10^{-3}$ is selected to obtain enough false-acceptance events for a useful Monte Carlo check. The network design still uses $10^{-4}$; the figure is not a claim to empirically certify a cryptographically small tail from a few hundred thousand trials. The invariance across waveform families and powers tests the conditioning argument in Theorem~\ref{thm:fa}.

	\subsection{Frequency Offset and AN Results}
	Figure~\ref{fig:cfo} uses one fixed, saved spherical template of length $32$. Each residual-offset point uses $100{,}000$ noise realizations. The exact prediction is conditional on that template and is evaluated by Poisson mixing of a noncentral-$F$ CDF over the orthogonal noncentrality. It agrees with the independently simulated distorted waveform. The zero-offset design rejects about $1\%$ of authorized blocks; a substantial phase rotation over the block increases rejection sharply. Raising nominal SINR without recalibrating the distorted direction is not a valid remedy.

	\begin{figure}[t]
		\centering
		\subfloat[False acceptance.]{
			\includegraphics[width=.47\columnwidth]{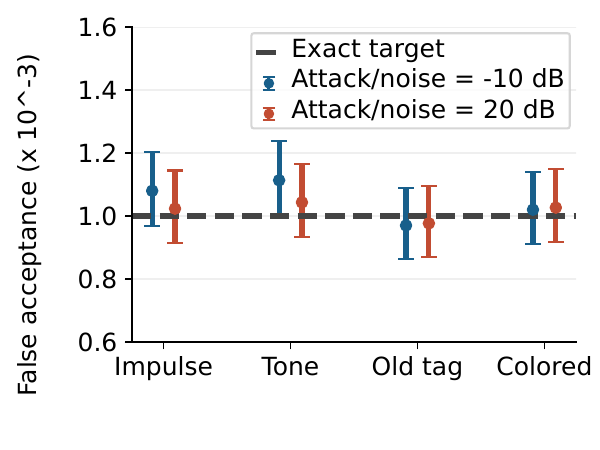}
			\label{fig:attacks}
		}
		\hfill
		\subfloat[Frequency-offset sensitivity.]{
			\includegraphics[width=.47\columnwidth]{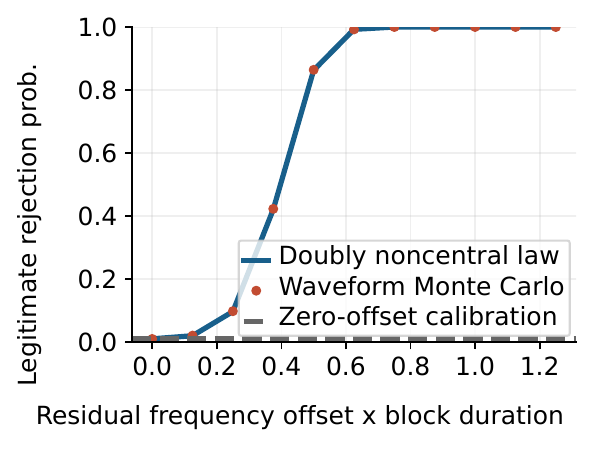}
			\label{fig:cfo}
		}
		\caption{%
			(a) False acceptance for tag-independent received attacks, with
			$n=32$ and target $10^{-3}$. Points and $95\%$ Wilson intervals use
			$300{,}000$ trials per case. The exact law is \eqref{eq:fa}. Blue and
			red circular markers show Monte Carlo estimates for attack-to-noise
			ratios of $-10$ and $20$ dB, respectively. The dashed line shows the
			analytical false-acceptance probability.
			(b) Residual common-frequency-offset sensitivity at the zero-offset
			calibrated tag energy. The horizontal axis is $\Delta f nT_s$. The
			prediction uses the actual template in \eqref{eq:doublebeta}, not a
			mean-coherence substitution. Red circular markers show waveform Monte
			Carlo estimates, while the solid blue curve shows the analytical
			prediction.%
		}
	\end{figure}

	Figure~\ref{fig:rank} is a reduced spectral experiment, not another optimized HAPS deployment. It uses six Eve dimensions. Two cases share AN eigenvalues $(0.3,0.6,1,1.4,0,0)$ but differ in whether the desired vector has energy in the last two coordinates. A third uses a full-rank covariance. Dotted high-power lines evaluate the first two terms of \eqref{eq:rank_law}. The rank-four exposed case grows linearly in power; the equally ranked covered case saturates. With an ideal interference-free legitimate SINR equal to the power scale, their secrecy benchmarks have qualitatively different behavior. The experiment isolates the geometric mechanism and does not establish a universal numerical antenna threshold.
	
	\begin{figure}[t]
		\centering
		\subfloat[Colluding SINR and high-power expansion.]{\includegraphics[width=.47\columnwidth]{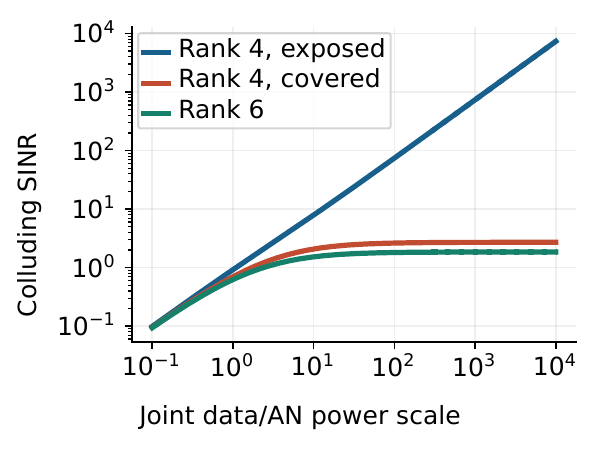}\label{fig:rank:a}}
		\hfill
		\subfloat[Corresponding secrecy-rate benchmark.]{\includegraphics[width=.47\columnwidth]{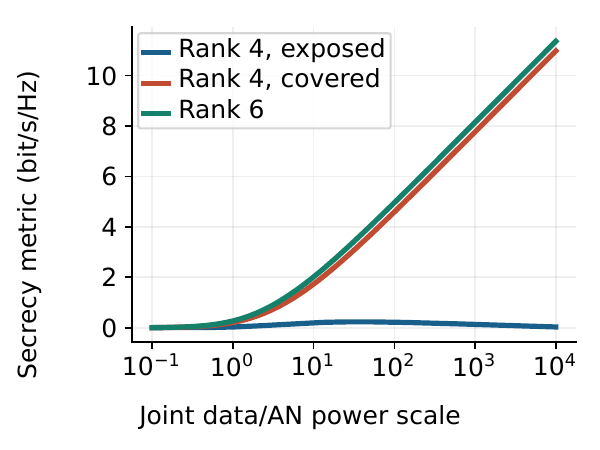}\label{fig:rank:b}}
		\caption{AN bottleneck in \eqref{eq:rank_law}. Panel (b) uses \eqref{eq:secrecy_metric} with $n=32$ and legitimate SINR equal to the power scale. Both rank-four cases have the same AN eigenvalues.}
		\label{fig:rank}
	\end{figure}
	
	\subsection{Network Power and Layout Results}
	Figure~\ref{fig:network} reports the full ten-snapshot comparisons.
	The mean joint-design powers for $n=16,32,64,128$ are
	$0.1931$, $0.1394$, $0.1183$, and $0.1089$ W, respectively.
	At $n=32$, the fixed-tag baseline uses $0.1921$ W on average,
	versus $0.1394$ W for joint design, a $27.44\%$ reduction in the ratio of paired-sample means.
	This is a transmit-power result at matched authentication and secrecy targets, not an increase in an authentication-tag secrecy proxy.
	All ten seeds are included at every plotted length; the right panel makes the seed dependence visible.
	
	For the 60 principal joint designs (four lengths at the default radius and two additional radii),
	the residual-corrected relative gaps range from $-0.00019\%$ to $0.01789\%$.
	The tiny negative endpoint reflects floating-point feasibility and eigenvalue tolerances, not a violation of the exact bound.
	The largest evaluated legitimate rejection probability is $0.0100001$, the smallest secrecy metric is $0.4999947$ bit/s/Hz,
	and the largest Eve SINR is $0.3000010$.
	The code uses an absolute $2\times10^{-5}$ feasibility allowance for rates, Eve SINR, and per-HAPS power, and the same relative allowance for the calibrated tag ratio.
	Conic solutions are obtained with CVXPY and Clarabel, with Splitting Conic Solver (SCS) used when the first solver reports numerical inaccuracy.
	
	In Fig.~\ref{fig:layout}, the best feasible choice from the three radii uses
	$0.1245$ W on average, $10.73\%$ below using the fixed $4$-km radius on these same snapshots.
	There is no universal best radius across seeds: geometry changes both the useful multicast directions and the spatial constraints imposed by colluding receivers.
	The three layouts are evaluated as deployment alternatives, not consecutive instantaneous flight actions.
	
	The Eve-cap sweep uses one explicitly identified snapshot, seed $100$.
	Allowing AN never needs to increase the exact optimum because zero AN remains feasible,
	but the actual benefit depends on the selected cap and spatial channels.
	At $\zeta=3$, the computed joint and no-AN powers are
	$0.7436$ and $0.8346$ W;
	at $\zeta=0.3$, the solutions use negligible AN and nearly identical power.
	These results do not support a universal percentage gain from adding AN.
	For reference, the much more restrictive fixed-beam design at seed $100$, $n=32$, $\zeta=0.3$ uses
	$13.4273$ W, illustrating why the optimized fixed-tag baseline is the more informative comparison.

	\begin{figure}[t]
		\centering
		\subfloat[Mean power versus authentication length.]{\includegraphics[width=.47\columnwidth]{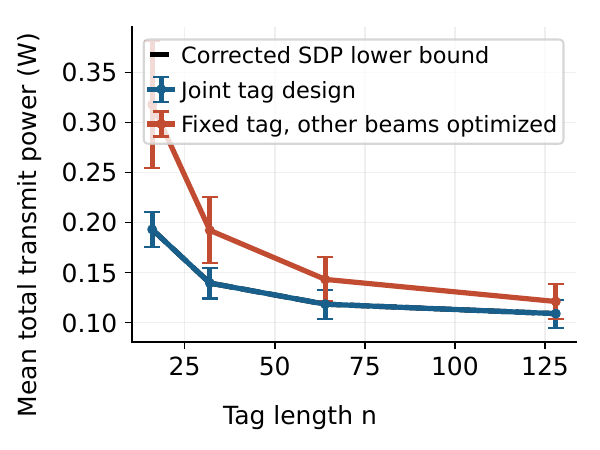}\label{fig:network:a}}
		\hfill
		\subfloat[Per-seed powers at $n=32$.]{\includegraphics[width=.47\columnwidth]{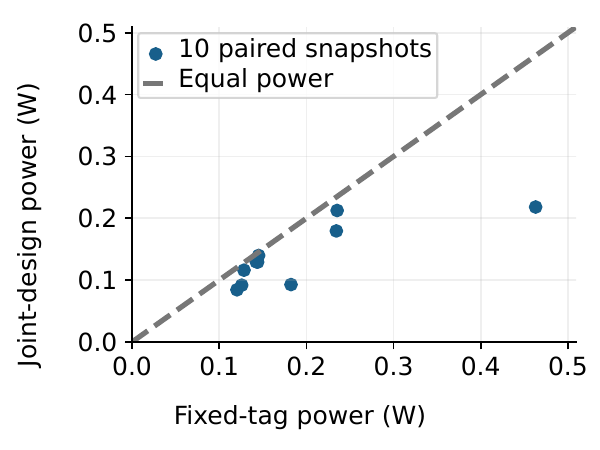}\label{fig:network:b}}
		\caption{Paired resource design across ten Rician snapshots. Error bars in panel (a) show standard errors across the ten snapshots, not confidence in global optimality. Panel (b) compares fixed-tag and joint-design powers for each seed. All plotted designs are checked against the original nonlinear constraints.}
		\label{fig:network}
	\end{figure}
	
	\begin{figure}[t]
		\centering
		\subfloat[Power across candidate HAPS ring radii.]{\includegraphics[width=.47\columnwidth]{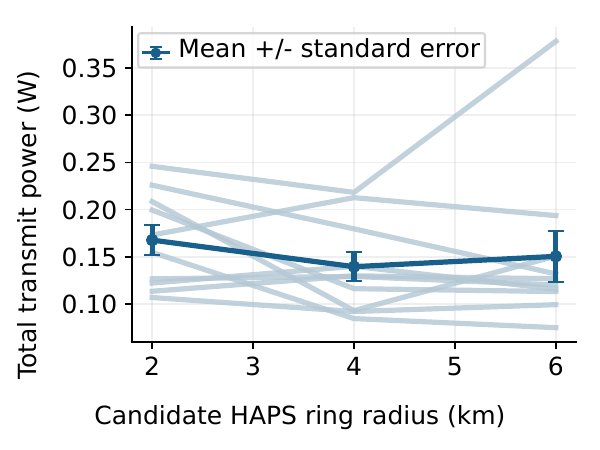}\label{fig:layout:a}}
		\hfill
		\subfloat[Joint and no-AN power versus Eve cap.]{\includegraphics[width=.47\columnwidth]{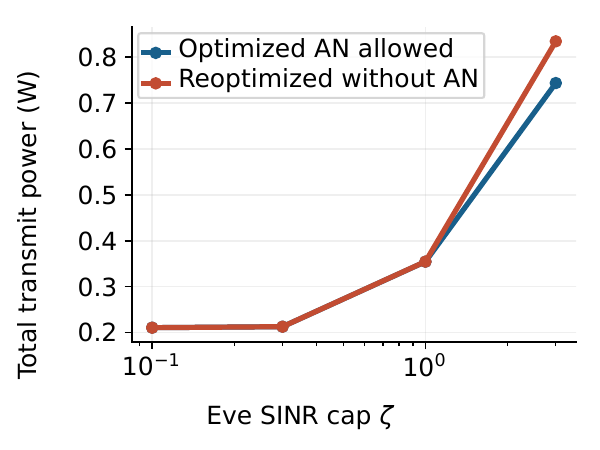}\label{fig:layout:b}}
		\caption{Deterministic design sensitivity. Panel (a) uses the same ten location/scattering draws and selects only within the finite layout catalog. Panel (b) shows the joint and no-AN solutions for seed $100$; changing the Eve cap also changes the sufficient legitimate target in \eqref{eq:Gamma}.}
		\label{fig:layout}
	\end{figure}
	
	\subsection{Discussion and Limitations}
	\label{sec:limitations}
	The results distinguish three resources with different roles. Fresh secret temporal direction controls blind impersonation probability. Received tag energy controls rejection of a legitimate response. Spatial AN controls the leakage directions available to cooperating receivers. More HAPS transmit power cannot replace a fresh response, and a low false-acceptance threshold cannot repair an uncovered Eve direction. These are operationally different constraints even when they share the same antennas and RF budget.
	
	The spherical tag is an analytical and waveform design choice. It is not a constant-envelope symbol sequence. A practical modem should quantify peak-to-average power, quantization of the pseudorandom expansion, and the effect of a finite private constellation. Our Gaussian receiver law is exact under its stated assumptions; labeling a simulation using a finite quadrature amplitude modulation (QAM) constellation as validation of that exact law without checking the approximation would be incorrect. The power-independent false-acceptance statement, in contrast, does not require the attacker or its receiver disturbance to be Gaussian, provided the current template remains independent.
	
	A second boundary is security timing. Previous-transcript replay is covered by independent freshness; forwarding the current response is not. A group insider knows the group key and is not an impersonator under the experiment. Channel estimates, synchronization pilots, fronthaul, and nonce management need their own trusted mechanisms. A broader protocol study would need explicit causal side-information and compromise models instead of extending Theorem~\ref{thm:fa} by assertion.
	
	A third boundary concerns spatial knowledge. Nulling and deterministic leakage caps depend on the CSI used for design. With uncertain or unknown Eve channels, this pointwise problem must be replaced by an uncertainty-set or outage formulation. Likewise, the ring catalog describes discrete deployment planning, not rapid HAPS repositioning over one tag block. The present calculations do not include flight-energy optimization or post-failure reassociation.
	
	Finally, numerical conic feasibility is checked with finite tolerances, and covariance lower bounds can be weak when dual residuals are appreciable. A recovery failure is not an infeasibility proof. Small reported gaps concern only the specified finite catalog, Eve cap, Gaussian channel model, and scheduled users. They do not demonstrate the uniqueness of a physical solution or the optimality of a real hardware implementation.
	
	\section{Conclusion}
	This paper connected finite-block group-tag verification to confidential beamforming in coordinated HAPS networks. Exact authentication laws determine the required tag SINR, while additional analysis characterizes the tag-energy floor, frequency-offset sensitivity, and peak-control penalty. Projection removes the tag with an explicit temporal-dimension cost, and spectral analysis identifies the AN coverage condition for bounded colluding-eavesdropper SINR. The minimum-power design combines SDR, exact private-beam rank recovery, and feasible refinement. The guarantees require fresh secret tags and tag-independent attacks. Joint beam design reduced mean transmit power by 27.4\% relative to the optimized fixed-tag baseline.

	\appendices
	\section{Short Derivations of Rate and Fixed-Shape Constraints}
	\label{app:derivations}
	
	\subsection{Projection loss and sufficient rate target}
	\label{app:projection-rate}
	Because $\mathbf B_s\HH\mathbf s=\mathbf0$, projection removes the tag exactly. The semi-unitarity relation $\mathbf B_s\HH\mathbf B_s=\mathbf I_{n-1}$ preserves the covariance of the temporally white Gaussian disturbance, so the retained dimensions have the SINR in \eqref{eq:legit_sinr}. Only $n-1$ of the $n$ dimensions remain, giving $c_n=(n-1)/n$. For $r_u>0$, the secrecy-rate requirement is
	\begin{align*}
		c_n\log_2(1+\gamma_u)-\log_2(1+\gamma_{e,u})\ge r_u.
	\end{align*}
	Since $\gamma_{e,u}\le\zeta$, imposing $c_n\log_2(1+\gamma_u)\ge r_u+\log_2(1+\zeta)$ is sufficient: subtracting $\log_2(1+\gamma_{e,u})$ leaves a value at least $r_u$. This sufficient condition rearranges to $\gamma_u\ge\Gamma_u=2^{[r_u+\log_2(1+\zeta)]/c_n}-1$ as in \eqref{eq:Gamma}. It is not, in general, necessary for the original secrecy-rate requirement.
	
	\subsection{Fixed-beam powers and the AN leakage floor}
	\label{app:fixed-beam}
	For the fixed-shape subproblem with the sufficient constraints $\gamma_u\ge\Gamma(n,\zeta)$ and $\gamma_{e,u}\le\zeta$, consider one user and null-space AN. The projected SINR is $\gamma_u=p_b g_b/\sigma_u^2$, so $\gamma_u\ge\Gamma$ gives $p_b^\star=\Gamma\sigma_u^2/g_b$. The authentication constraint then gives $p_t g_t/(p_b^\star g_b+\sigma_u^2)\ge\rho_n^\star$, hence the minimum $p_t^\star$ in \eqref{eq:fixed_powers}. Larger data or tag powers cannot improve this fixed-shape minimum-power objective or relax the Eve cap.
	
	For the AN feasibility condition, diagonalize $\mathbf A_0=\mathbf U\operatorname{diag}(a_i)\mathbf U\HH$ and set $z_i=[\mathbf U\HH\mathbf f_0]_i$. Then
	\begin{align*}
		&p_b^\star\mathbf f_0\HH(p_a\mathbf A_0+\sigma_e^2\mathbf I)^{-1}\mathbf f_0\\
		&\quad=p_b^\star\sum_{i:a_i>0}\frac{|z_i|^2}{p_a a_i+\sigma_e^2}
		+\frac{p_b^\star}{\sigma_e^2}\sum_{i:a_i=0}|z_i|^2.
	\end{align*}
	The positive-eigenvalue terms vanish as $p_a\to\infty$, while the null-space term stays fixed. Thus the limit is the leakage floor after \eqref{eq:an_scalar}; monotonicity gives the stated feasibility conditions.
	
	\subsection{Rank-one form of the Eve matrix constraint}
	\label{app:eve-lmi}
	For $\mathbf X_u=\mathbf v_u\mathbf v_u\HH$ and $\mathbf C_e\succ\mathbf0$, congruence of \eqref{eq:sdr_eve} by $\mathbf C_e^{-1/2}$ gives
	\[
	(\mathbf C_e^{-1/2}\mathbf G\mathbf v_u)(\mathbf C_e^{-1/2}\mathbf G\mathbf v_u)\HH\preceq\zeta\mathbf I.
	\]
	The left side is rank one, with sole possible nonzero eigenvalue $\mathbf v_u\HH\mathbf G\HH\mathbf C_e^{-1}\mathbf G\mathbf v_u=\gamma_{e,u}$. The matrix inequality is therefore equivalent to the scalar Eve cap for rank-one private covariance; for unrestricted $\mathbf X_u\succeq\mathbf0$, it remains a convex relaxation.
	
	\IfFileExists{refs.bib}{
		\bibliographystyle{IEEEtran}
		\bibliography{refs}
	}{
	}
	
\end{document}